\documentclass[final]{IEEEtran}
\usepackage{tikz}
\usepackage{amsmath,amsfonts,amssymb}%
\usepackage{bm}%
\usepackage{graphicx,graphics}%
\usepackage{cases}%
\usepackage[noadjust]{cite}%
\usepackage{color}%
\usepackage{cite,url}
\usepackage{verbatim}
\usepackage{algorithm}
\usepackage{balance}
\usepackage{multirow}
\usepackage{stfloats}
\usepackage{xspace}
\usepackage{amsthm}
\usepackage{mathtools} 

\usepackage{caption}

\usepackage{tikz}
\usetikzlibrary{shapes.misc}
\usetikzlibrary{matrix}
\usetikzlibrary{arrows,backgrounds,fit,calc}

\usepackage{float}
\floatname{algorithm}{Procedure}

\newcommand{\bv}[1]{\boldsymbol{#1}}
\newcommand{\dfn}{\triangleq}
\newcommand{\untsph}{\mathbb{S}^{2}} 

\newcommand{\shc}[3]{({#1})_{#2}^{#3}}
\newcommand{\lsph}{L^2(\untsph)}
\newcommand{\unit}[1]{\bv{\hat{#1}}}

\newcommand{\figref}[1]{Fig.\,\ref{#1}}

\newcommand{\secref}[1]{Section\,\ref{#1}}

\newcommand{\matlab}{\texttt{MATLAB}}

\newcommand{\nspace}{\mathcal{N}}

\DeclarePairedDelimiterX\innerp[2]{\langle}{\rangle}{#1,#2}

\newtheorem{lemma}{Lemma}

\newtheorem{remark}{Remark}

\graphicspath{{figs/},{pdfs/},{pdf/},{Figures/},{Authors/}}

\begin{document}
\title{Accurate Reconstruction of Finite Rate of Innovation Signals on the Sphere}
%
\author{%
Yahya Sattar, Zubair~Khalid,~\IEEEmembership{Member,~IEEE}, and Rodney A. Kennedy,~\IEEEmembership{Fellow,~IEEE}
 \thanks{Y.~Sattar and Z.~Khalid are with the School of Science and Engineering, Lahore University of Management Sciences, Lahore, Pakistan. R.~A.~Kennedy is with
   the Research School of Engineering, College of Engineering and
   Computer Science, The Australian National University, Canberra,
}
\thanks{Z.~Khalid and R.~A.~Kennedy are supported by the Australian Research Council's Discovery Projects funding scheme (Project no. DP150101011)
}

\thanks{E-mail: yahya.sattar@lums.edu.pk, zubair.khalid@lums.edu.pk, rodney.kennedy@anu.edu.au}
}
\maketitle
\begin{abstract}

We develop a method for the accurate reconstruction of non-bandlimited finite rate of innovation signals on the sphere.
For signals consisting of a finite number of Dirac functions on the sphere, we develop an annihilating filter based method for the accurate recovery of parameters of the Dirac functions using a finite number of observations of the bandlimited signal. In comparison to existing techniques, the proposed method enables more accurate reconstruction primarily due to better conditioning of systems involved in the recovery of parameters. For the recovery of $K$ Diracs on the sphere, the proposed method requires samples of the signal bandlimited in the spherical harmonic~(SH) domain at SH degree equal or greater than $ K + \sqrt{K + \frac{1}{4}} - \frac{1}{2}$. In comparison to the existing state-of-the art technique, the required bandlimit, and consequently the number of samples, of the proposed method is the same or less. We also conduct numerical experiments to demonstrate that the proposed technique is more accurate than the existing methods by a factor of $10^{7}$ or more for $2 \le K\le 20$.
\end{abstract}
\begin{IEEEkeywords}
unit sphere, sampling, finite rate of innovation, signal reconstruction, spherical harmonic transform.
\\[1mm]\centering EDICS: DSP-FRI, DSP-SPARSE.
\end{IEEEkeywords}

\section{Introduction}

Development of spherical signal processing techniques finds direct applications in diverse fields of science and engineering where signals are naturally defined on the sphere. These applications include, but not limited to, spherical harmonic lighting in computer graphics~\cite{Zhang:2006}, signal analysis in diffusion magnetic resonance imaging~(dMRI)~\cite{Bates:2016,Samuel:2012}, spectrum estimation in geophysics and cosmology~\cite{Dahlen:2008}, sound analysis and reproduction in acoustics~\cite{Bates:2015,Dokmanic:2016} and placement of antennas in wireless communication~\cite{Alem:2015}. To support signal analysis in these applications, accurate reconstruction of signals from a finite number of measurements is inherently required and is therefore of significant importance. In this work, we consider the problem of sampling and accurate reconstruction of non-bandlimited finite rate of innovation~(FRI) signal, consisting of finite $K$ number of Dirac delta functions, on the sphere.

Many sampling schemes have been devised in the literature~(e.g., see \cite{Khalid:2014} and references therein) for the acquisition of signals bandlimited in the spectral domain, which is enabled by the spherical harmonic~(SH) transform -- a natural counterpart of the Fourier transform for signals on the sphere~\cite{Kennedy-book:2013}. For the accurate computation of SH transform and accurate reconstruction of a signal bandlimited at SH degree $L$~(formally defined in \secref{subsec:MB}), we require $L^2$ number of samples~\cite{Khalid:2014}. The sampling schemes for taking measurements of bandlimited signals, although permit accurate reconstruction of signal, are not suitable for sampling of non-bandlimited signals such as an ensemble of spikes~(Dirac delta functions in the limit) on the sphere which appear in applications in dMRI~\cite{Samuel:2012}, acoustics and cosmology~\cite{Dokmanic:2016}.

Based on the super-resolution theory~\cite{Candes:2014}, an algorithm has been developed in \cite{Bendory:2015} for the reconstruction of FRI signals using semi-definite programming, root finding and least-squares. However, the method is iterative in nature and requires Dirac functions to satisfy a minimum separation condition. Recently, following the annihilating filter method devised for signals in one-dimensional Euclidean domain~\cite{Vetterli:2002} and extended to 2D and higher dimensions~\cite{Shukla:2007}, signal processing techniques have been proposed~\cite{Sameul:2013,Dokmanic:2016} for the recovery of parameters of FRI signal on the sphere. The method proposed in \cite{Sameul:2013} requires bandlimiting the FRI signal at $L=2K$ for the recovery of parameters of Diracs. To reduce the total number of measurements, an alternative reconstruction technique has been developed in \cite{Dokmanic:2016} which requires the measurements of the FRI signal bandlimited at $L\ge(K+\sqrt{K})$ and therefore reduces the number of samples requirement by a factor of approximately four. For both of these schemes, the error in the reconstruction or recovery of parameters increases with the number of Diracs due to ill-conditioning of the systems required to be inverted during the recovery of parameters.

In this work, we also employ the annihilating filter method in order to develop a method for the recovery of parameters, and consequently accurate reconstruction, of an FRI signal composed of $K$ Diracs on the sphere. Our method requires bandlimiting the signal at $L\ge(K+\sqrt{K+ \frac{1}{4}} -\frac{1}{2})$ and therefore takes the same or less\footnote{Since bandlimit $L$ is required to be an integer, the difference between the bandlimit required by the best of existing algorithms and the proposed method is zero or differs by one.} number of samples compared to the existing methods. More importantly, in comparison to existing techniques, our method enables more accurate recovery of parameters of FRI signals as we demonstrate through numerical experiments. We organize the rest of the paper as follows. In next section we present the mathematical background and review the existing methods for the problem under consideration. Proposed method is developed in Section III and its analysis is carried out in Section IV. Finally, we make the concluding remarks Section V.

\section{Preliminaries and Problem Formulation} \label{sec:P&PF}

\subsection{Mathematical Background -- Signals on the Sphere}\label{subsec:MB}

The unit sphere or 2-sphere is defined as $\mathbb{S}^2 = \{\unit{u} \in \mathbb{R}^3 : |\unit{u}|_2 = 1\}$, where $|\cdot|_2$ denotes the Euclidean norm and $\unit{u}$ is the unit vector, parameterized  in terms of $\theta$ and $\phi$ as $\unit{u} \equiv \unit{u}(\theta,\phi) \triangleq (\sin\theta\,\cos\phi, \; \sin\theta\,\sin\phi, \; \cos\theta)^{\prime}$. Here $\theta \in [0,\,\pi]$ is the colatitude angle and $\phi \in [0, \, 2\pi)$ is the longitude angle. The inner product between two functions $f$ and $g$ on the sphere is defined as
\begin{align}\label{eqn:innprd}
\langle f, g \rangle \triangleq  \int_{\mathbb{S}^2} f(\unit{u}) \,\overline {g(\unit{u})} \,ds(\unit{u}),
\end{align}
where $ds(\unit{u}) = \sin\theta\, d\theta\,d\phi$ is the differential area element on $\mathbb{S}^2$, $\overline{(\cdot)}$ denotes the complex conjugate and the integration is carried out over the entire sphere. The complex-valued functions on the 2-sphere form a Hilbert space $\lsph$ equipped with the inner product defined in $\eqref{eqn:innprd}$. The inner product defined in $\eqref{eqn:innprd}$ induces a norm $\|f\| \triangleq {\langle f,f \rangle}^{1/2}$. We refer to the functions with finite induced norm as signals on the sphere.

For the space $\lsph$, spherical harmonic functions~(or spherical harmonics for short) serve as complete orthonormal basis and are defined as~\cite{Kennedy-book:2013}
\begin{equation}
\label{eqn:YLMrep}
Y_{\ell}^{m}(\unit{u})\equiv Y_{\ell}^{m}(\theta,
\phi) \dfn
\sqrt{\frac{2{\ell}+1}{4\pi}\frac{({\ell}-m)!}{({\ell}+m)!}}\,
        P_{\ell}^{m}(\cos\theta)e^{im\phi},
\end{equation}
for integer degree $\ell \geq 0$ and integer order $|m| \leq \ell$. Here, $P_{\ell}^{m}$ is the associated Legendre function of degree $\ell$ and order $m$ and is given by~\cite{Kennedy-book:2013}
\begin{align}
\label{eqn:AssocLegendre}
    P_{\ell}^{m}(\nu)
        &= \frac{(-1)^m}{2^{\ell} {\ell}!} (1-\nu^2)^{m/2} \frac{d^{{\ell}+m}}{d\nu^{{\ell}+m}}
            (\nu^2-1)^{\ell} \\
    P_{\ell}^{-m}(\nu)
        &= {(-1)^m} \frac{({\ell}-m)!}{({\ell}+m)!} P_{\ell}^m(\nu),
\end{align}
for $0 \leq m \leq \ell$ and $|\nu|\leq 1$. Due to completeness of spherical harmonics, we can represent any signal $f\in\lsph$ as
\begin{align}\label{eqn:HarmonExpan}
f(\unit{u}) = \sum\limits_{\ell = 0}^{\infty}\sum\limits_{m = -\ell}^{\ell} \shc{f}{\ell}{m} \, Y_{\ell}^m(\unit{u}),
\end{align}
where $\shc{f}{\ell}{m} \dfn \langle f, Y_{\ell}^m \rangle$~\cite{Kennedy-book:2013} denotes the SH coefficient of integer degree $\ell \geq 0$ and integer order $|m| \leq \ell$. The spherical harmonic coefficients form the representation of a signal in spectral~(Fourier) domain. We define the function $f$ to be bandlimited in spectral domain at degree $L$ if $\shc{f}{\ell}{m}=0,\,\forall\, \ell \ge L,\, -\ell \leq m \leq \ell$.
%
\subsection{Problem under Consideration}
We consider a signal consisting of $K$ Diracs on the sphere given by
\begin{align}\label{eqn:diracs}
f(\unit{u}) = \sum\limits_{k =1}^K \alpha_k \,\delta(\unit{u},\unit{u}_k),
\end{align}
where $\alpha_k$ is the complex amplitude and $\unit{u}_k \equiv \unit{u}_k(\theta_k,\phi_k)$ represents the location of $k$-th Dirac on the sphere. Here $\delta(\unit{u},\unit{u}_k)$ is the Dirac delta function defined on the sphere which, similar to its linear counterpart, is identified by its sifting property $\langle f, \delta(\cdot\, , \unit{u}_k)\rangle = f(\unit{u}_k)$. The problem under consideration is to accurately recover the amplitudes $\alpha_k$ and locations $\unit{u}_k$ of $K$ Diracs of the signal $f$, given the samples of $f$ bandlimited in the spectral domain.

\subsection{Review of Existing Methods}

We here review the existing methods presented in literature~\cite{Sameul:2013,Dokmanic:2016} for the recovery of parameters of signal of the form given in \eqref{eqn:diracs}. Recently, an algorithm has been presented in \cite{Sameul:2013}, based on the annihilation filter method~\cite{Vetterli:2002}, for the recovery of the parameters $f$ which requires the computation of spherical harmonic coefficients $\shc{f}{\ell}{m}$ of $f$ for degrees $\ell<2K$ and orders $|m|\le\ell$, which are computed by first convolving the signal $f$ with a sampling kernel which bandlimits the signal at degree $L = 2K$. If the recently proposed optimal-dimensionality sampling~\cite{Khalid:2014} is employed for the computation of SH coefficients, the method requires $L^2$ samples of the signal $f$ bandlimited at $L=2K$. Employing the SH coefficients, the method then forms a Toeplitz system which enables the computation of $\phi_k$ using which $\alpha_k$ and $\theta_k$ are recovered. The method assumes that $\theta_k \notin \{0, \pi\}$ and $\theta_j \neq \pi - \theta_k$ when $\phi_j = \phi_k$ for $j,k = 1,2,\hdots,K$ and $j\ne k$.

To reduce the number of samples required for the recovery of parameters, an algorithm has been presented more recently in \cite{Dokmanic:2016}, which requires SH coefficients $\shc{f}{\ell}{m}$ for $\ell \le L$ and orders $|m|\le\ell$ of the signal $f$ bandlimited at SH degree\footnote{Here $\lceil \cdot\rceil$ denotes the integer ceiling function.} $L = \lceil K + \sqrt{K} \rceil$.
Consequently, when compared to the method in \cite{Sameul:2013}, this method requires~(approximately) \emph{four} times less number of samples. The SH coefficients are then used to form an annihilating matrix~\cite{Vetterli:2002} which enables the computation of $\theta_k$, which are then used to recover the parameters $\alpha_k$ and $\phi_k$. The algorithm works only when $\theta_k \notin \{0, \pi\}$ and $\theta_j \neq \theta_k$ for $j \neq k$ and $j,k = 1,2,\hdots,K$. Although both of these methods allow recovery of parameters, the reconstruction error increases with the increase in number of Diracs on the sphere as we illustrate later in the paper.

\section{Accurate Reconstruction of Signals with FRI}\label{sec:ARSF}
Here, we propose a method for the recovery of the parameters of $f$ given in \eqref{eqn:diracs}. In comparison to the existing methods~\cite{Sameul:2013,Dokmanic:2016}, the proposed algorithm has significantly smaller reconstruction error. Furthermore, our method requires (approximately)~four times less and the same~(or less) number of samples than those required by the methods presented in \cite{Sameul:2013} and \cite{Dokmanic:2016}, respectively.

\subsection{Formulation}

By employing the sifting property of Dirac delta function, we can express the SH coefficient $\shc{f}{\ell}{m} = \langle f, Y_{\ell}^m \rangle$ of $f$ given in \eqref{eqn:diracs} as
\begin{align}\label{eqn:flmBysift}
\shc{f}{\ell}{m} = \sum\limits_{k =1}^K \alpha_k \, \overline{Y_{\ell}^m(\theta_k, \phi_k)}.
\end{align}
which, noting that $Y_{\ell}^m(\theta, \phi) = Y_{\ell}^m(\theta, 0)\, e^{im\phi}$~(using \eqref{eqn:YLMrep}) and $Y_\ell^m(\theta,0)$ is a product of $(\sin\theta)^{|m|}$ and a polynomial in $\cos\theta$ of degree $(\ell-|m|)$~(using \eqref{eqn:AssocLegendre}), can be expressed as
\begin{align}\label{eqn:flmPolynomial}
\shc{f}{\ell}{m} = \sum\limits_{k =1}^K \alpha_k \sum\limits_{p=0}^{\ell-|m|} c_{\ell m}^p\,(\cos\theta_k)^p (\sin\theta_k)^{|m|} e^{-im\phi_k},
\end{align}
where $c_{\ell m}^p$ denotes the coefficient associated with $(\cos\theta)^p$ of the polynomial defining $Y_\ell^m(\theta,0)$. 
%
For order $m$, we rearrange $\eqref{eqn:flmPolynomial}$ to get,
\begin{align}\label{eqn:FLM_2}
\shc{f}{\ell}{m} = \sum_{p=0}^{\ell-|m|} {c}_{\ell m}^{p} d_{pm},
\end{align}
where
\begin{align}\label{eqn:dpm}
d_{pm} = \begin{cases}
\sum\limits_{k = 1}^K (\alpha_k y_{kp})x_k^m &\quad 0\le m <L ,\\
\sum\limits_{k = 1}^K   (\alpha_k y_{kp})    \overline{x_k^m} &\quad -L< m <0 ,
\end{cases}
\end{align}
with $x_k = \sin\theta_k e^{-i\phi_k}$ and $y_{kp} = (\cos\theta_k)^p$. Clearly, both $d_{pm}$ for $0\le m <L$ and $\overline{d_{pm}}$ for $-L< m <0$ are linear combination of exponentials $x_k^m$ and therefore are of special interest as the annihilating filter technique~\cite{Vetterli:2002} can be used to recover $x_k$.
\subsection{Recovery of Longitudes of Diracs}

We consider that the measurements of the signal $f$ ban-limited at degree $L$ are available such that the spherical harmonic coefficients $\shc{f}{\ell}{m}$ can be accurately computed for all degrees $\ell <L$ and all orders $|m|\le\ell$.
We shortly present the bandlimit $L$ required for the accurate recovery of parameters.

In \eqref{eqn:FLM_2}, $\shc{f}{\ell}{m}$ for $|m|\le\ell<L$ and $d_{pm}$ for $0 \le p <L-|m|$ form a linear system of equations for each $ |m|<L$ with triangular coefficient matrix of size $(L-|m|)\times(L-|m|)$. Consequently, $d_{pm}$ for $0 \le p <L-|m|$ can be recovered {\em exactly} for each $ |m| <L$ using \eqref{eqn:FLM_2}.

Once $d_{pm}$ is computed, we employ the annihilating filter technique~\cite{Vetterli:2002} to estimate $x_k$ as $d_{pm}$ is a linear combination of $K$ powers of $x_k$. This estimation involves the construction of annihilating matrix $\mathbf{Z}$ given by
\begin{align}\label{eqn:Zmat1}
\mathbf{Z} =
\begin{bmatrix}
d_{0,L-1} & d_{0,L-2} & \cdots & d_{0,L-K-1}\\
d_{0,L-2} & d_{0,L-3} & \cdots & d_{0,L-K-2}\\
\vdots & \vdots & \ddots & \vdots\\
d_{0,K} & d_{0,K-1} & \cdots & d_{0,0}\\
\overline{d_{0,-(L-1)}} & \overline{d_{0,-(L-2)}} & \cdots & \overline{d_{0,-(L-K-1)}}\\
\overline{d_{0,-(L-2)}} & \overline{d_{0,-(L-3)}} & \cdots & \overline{d_{0,-(L-K-2)}}\\
\vdots & \vdots & \ddots & \vdots\\
\overline{d_{0,-(K)}} & \overline{d_{0,-(K-1)}} & \cdots & \overline{d_{0,0}}\\
d_{1,L-2} & d_{1,L-3} & \cdots & d_{1,L-K-2}\\
d_{1,L-3} & d_{1,L-4} & \cdots & d_{1,L-K-3}\\
\vdots & \vdots & \ddots & \vdots\\
d_{1,K} & d_{1,K-1} & \cdots & d_{1,0}\\
\overline{d_{1,-(L-2)}} & \overline{d_{1,-(L-3)}} & \cdots & \overline{d_{1,-(L-K-2)}}\\
\overline{d_{1,-(L-3)}} & \overline{d_{1,-(L-4)}} & \cdots & \overline{d_{1,-(L-K-3)}}\\
\vdots & \vdots & \ddots & \vdots
\end{bmatrix},
\end{align}
followed by the computation of its right singular vector $\mathbf{v}$.
\begin{lemma}
If the longitudes and colatitudes of $K$ Diracs of the signal $f$ are such that $\theta_j \neq \pi - \theta_k$ when $\phi_j = \phi_k$ for $j,k = 1,2,\hdots,K$, $j\ne k$, the null-space, denoted by $\nspace(\mathbf{Z})$, of the annihilating matrix $\mathbf{Z}$ with at least $K$ rows is $1$-dimensional.
\end{lemma}
\begin{proof}
If $\theta_j \neq \pi - \theta_k$ when $\phi_j = \phi_k$ for $j,k = 1,2,\hdots,K$, $j\ne k$, each $x_k = \sin\theta_k e^{-i \phi_k}$ corresponding to the $k$-th Dirac is unique. From $\eqref{eqn:dpm}$, we note that the rank of $\mathbf{Z}$ with at least $K$ rows is equal to $K$ when all $x_k$ are unique. Consequently, $\nspace(\mathbf{Z})$ is $1$-dimensional.
\end{proof}
\begin{remark}[On the Recovery of $x_k$]
Now we employ the annihilating filter property to estimate $x_k$ using $\mathbf{v}\in\nspace(\mathbf{Z})$. A finite impulse response~(FIR) filter is known as annihilating filter if zeros of the filter are placed such that the filter annihilates the signal. Since $\mathbf{v} \in \nspace(\mathbf{Z})$, we have
\begin{align} \label{eqn:rowOfZv}
\mathbf{d}_q^{\mathsf{T}}\mathbf{v} = 0,
\end{align}
for any $\mathbf{d}_q \dfn \mathbf{Z}\{q, :\}$, that is, the $q$-th row of the annihilating matrix $\mathbf{Z}$. Consequently, $\mathbf{v}$ is a vector of coefficients of the FIR filter which annihilates the signal of the form $\eqref{eqn:dpm}$. The transfer function of such annihilating filter is given by
\begin{align}\label{eqn:Vofz}
V(z) \dfn \prod\limits_{k = 1}^K (1 - x_k z^{-1}) \dfn \sum\limits_{n= 0}^K v_n z^{-1}.
\end{align}
Since we have determined $\mathbf{v}$, we obtain the estimate of $x_k,\,k=1,2,\hdots,K$ by finding the roots, denoted by $\tilde{x}_k,\,k = 1,2,\hdots,K$, of the annihilating filter.
\end{remark}

Using $\tilde{x}_k$ and noting that $x_k = \sin\theta_k e^{-i \phi_k}$, we recover the longitudes $\phi_k$ from $\tilde{x}_k$ as
\begin{align}\label{eqn:recOfphi}
\tilde{\phi}_k = -{\rm Phase}\{\tilde{x}_k\},
\end{align}
where ${\rm Phase}\{\cdot\}$ returns the phase of the complex argument.

\begin{remark}[On the Bandlimit Requirement]
For a signal $f$ bandlimited at $L$, the maximum number of rows of $\mathbf{Z}$ which can be constructed is $2\times(L-K)+2\times(L-K-1)+\cdots+2\times2+2\times1$. Following Lemma 1, we require matrix $\mathbf{Z}$ to have at least $K$ rows to ensure a unique vector $\mathbf{v} \in \nspace(\mathbf{Z})$. Consequently, we require
\begin{align}\label{eqn:boundOnL}
\begin{split}
L & \geq K + \sqrt{K + \frac{1}{4}} - \frac{1}{2}
\end{split}
\end{align}
\end{remark}

\subsection{Recovery of Colatitudes and Amplitudes of Diracs}

Here we use the estimated $\tilde x_k$ to recover colatitude $\theta_k$ and amplitude $\alpha_k$  for $k = 1,\hdots,K$.
For $p=0$, we have $d_{pm}$ for $m = 0,1,\hdots,L-1$, which we explicitly rewrite, using \eqref{eqn:dpm}, as
\begin{align}{\label{eqn:d0m}}
d_{0m} = \sum\limits_{k = 1}^K \alpha_k x_k^m.
\end{align}
Since we have chosen $L > K$~(Remark 2), we can form the following Vandermonde system using $\eqref{eqn:dpm}$:
\begin{align}\label{eqn:VandforAlpha}
\begin{bmatrix}
1 &1 &\cdots &1\\
x_1 &x_2 &\cdots &x_k\\
\vdots &\vdots &\ddots &\vdots\\
x_1^{K-1} &x_2^{K-1} &\cdots &x_k^{K-1}
\end{bmatrix}
\begin{bmatrix}
\alpha_1\\
\alpha_2\\
\vdots\\
\alpha_K
\end{bmatrix}
=
\begin{bmatrix}
d_{00}\\
d_{01}\\
\vdots\\
d_{0K-1}
\end{bmatrix}.
\end{align}
Provided $x_k, \; k = 1,2,\hdots,K$, are distinct as ensured by Lemma 1, the Vandermonde system above enables recovery of amplitudes $\tilde{\alpha}_k, \; k = 1,2,\hdots,K$.

To recover colatitude parameter, we use $d_{pm}$ for $p=1$ and $m = 0,1,\hdots,L-2$. Using $d_{1m}$, given in $\eqref{eqn:dpm}$, and noting that $y_{k1} = \cos\theta_k$, we formulate another Vandermonde system given by
\begin{align}\label{eqn:VandforTheta}
\begin{bmatrix}
1 &1 &\cdots &1\\
x_1 &x_2 &\cdots &x_k\\
\vdots &\vdots &\ddots &\vdots\\
x_1^{K-1} &x_2^{K-1} &\cdots &x_k^{K-1}
\end{bmatrix}
\begin{bmatrix}
\alpha_1 \cos\theta_1\\
\alpha_2 \cos\theta_2\\
\vdots\\
\alpha_K \cos\theta_K
\end{bmatrix}
=
\begin{bmatrix}
d_{10}\\
d_{11}\\
\vdots\\
d_{1K-1}
\end{bmatrix}.
\end{align}
The solution of above system yields an estimate of $\alpha_k \cos\theta_k$, denoted by ${\mathbb{E}}\{\alpha_k \cos\theta_k\}$, for $k = 1,2,\hdots,K$. Since we have already recovered amplitude as $\tilde{\alpha}_k$, we recover the colatitude as
\begin{align}\label{eqn:recOfTheta}
\tilde{\theta_k} = \arccos \bigg[\frac{ \mathbb{E}\{\alpha_k \cos\theta_k \}  }{\tilde{\alpha}_k} \bigg], \; k = 1,2,\hdots,K.
\end{align}

\section{Analysis}\label{sec:illust}
Here we compare the bandlimit requirement and reconstruction accuracy of the proposed method with the existing methods~\cite{Sameul:2013,Dokmanic:2016}. We corroborate, through numerical experiments, the claim that the proposed method is superior in terms of accuracy of recovered parameters in comparison with the existing methods.

\subsection{Bandlimit Requirement}

As mentioned earlier, we require $L^2$ number of measurements of the signal bandlimited at $L$ to compute its spherical harmonic coefficients. Consequently, it is desirable for a method to have smaller bandlimit requirements to reduce the total number of measurements. For the recovery of parameters of the signal consisting of $K$ Diracs on the sphere, the bandlimit required by the proposed algorithm is $L  =\lceil K + \sqrt{K + \frac{1}{4}} - \frac{1}{2}\rceil$, which is much smaller as compared to $L=2K$ required by the method in \cite{Sameul:2013} and the same~(or \emph{less}) than  $L = \lceil K + \sqrt{K} \rceil$ required for the algorithm presented in \cite{Dokmanic:2016}.

\subsection{Accuracy Analysis}
In order to compare the recovery/reconstruction error of the proposed method with the algorithms presented in \cite{Sameul:2013} and \cite{Dokmanic:2016}, we implement each method in \matlab\footnote{For our method, we make the code publicly available at \url{http://zubairkhalid.org/fri} to facilitate the reproduction of research results. }\, and recover the parameters of the signal $f$, of the form given in \eqref{eqn:diracs}, by conducting following experiment. For each $K = 2,4,\hdots,20$, we randomly choose the parameters\footnote{The parameters are randomly generated such that we have $K$ distinct $\theta_k$ and $x_k = \sin\theta_k e^{-i\phi_k}$ as the method in \cite{Dokmanic:2016} requires $\theta_k,\,k=1,2,\hdots,K$ to be unique, whereas the method in \cite{Sameul:2013} and our proposed method require $x_k,\,k=1,2,\hdots,K$ to be unique~(Lemma 1). This is avoided by imposing the condition that the $K$ Diracs have at least $\pi/3K$ distance among them.} $\theta_k\in(0,\pi)$, $\phi_k \in [0,2\pi)$ and $\alpha_k$ with real and imaginary parts uniformly distributed in $[-1,1]$ for $k=1,2,\hdots,K$. For each method and each $K$, we recover the parameters $\tilde{\theta}_k$, $\tilde{\phi}_k$ and $\tilde{\alpha}_k$ of the signal and compute the mean-squared errors given by
$
E_{\theta} = \frac{1}{K} \sum\limits_{k=1}^{K} |\tilde{\theta}_k-\theta_k|^2$,  $E_{\phi} = \frac{1}{K} \sum\limits_{k=1}^{K} |\tilde{\phi}_k-\phi_k|^2$, $E_{\alpha} = \frac{1}{K} \sum\limits_{k=1}^{K} |\tilde{\alpha}_k-\alpha_k|^2$.
We plot these errors, averaged over $1000$ trials of the experiment, in \figref{fig:errors}, where it is evident that the proposed method enables more accurate recovery of parameters when compared to other methods in literature. On average, the proposed algorithm outperforms the other methods in terms of smaller recovery error by a factor up to $10^{7}$.

\begin{figure}[!t]
    \centering
    \includegraphics[scale=0.65,trim={0 0  0 0 }]{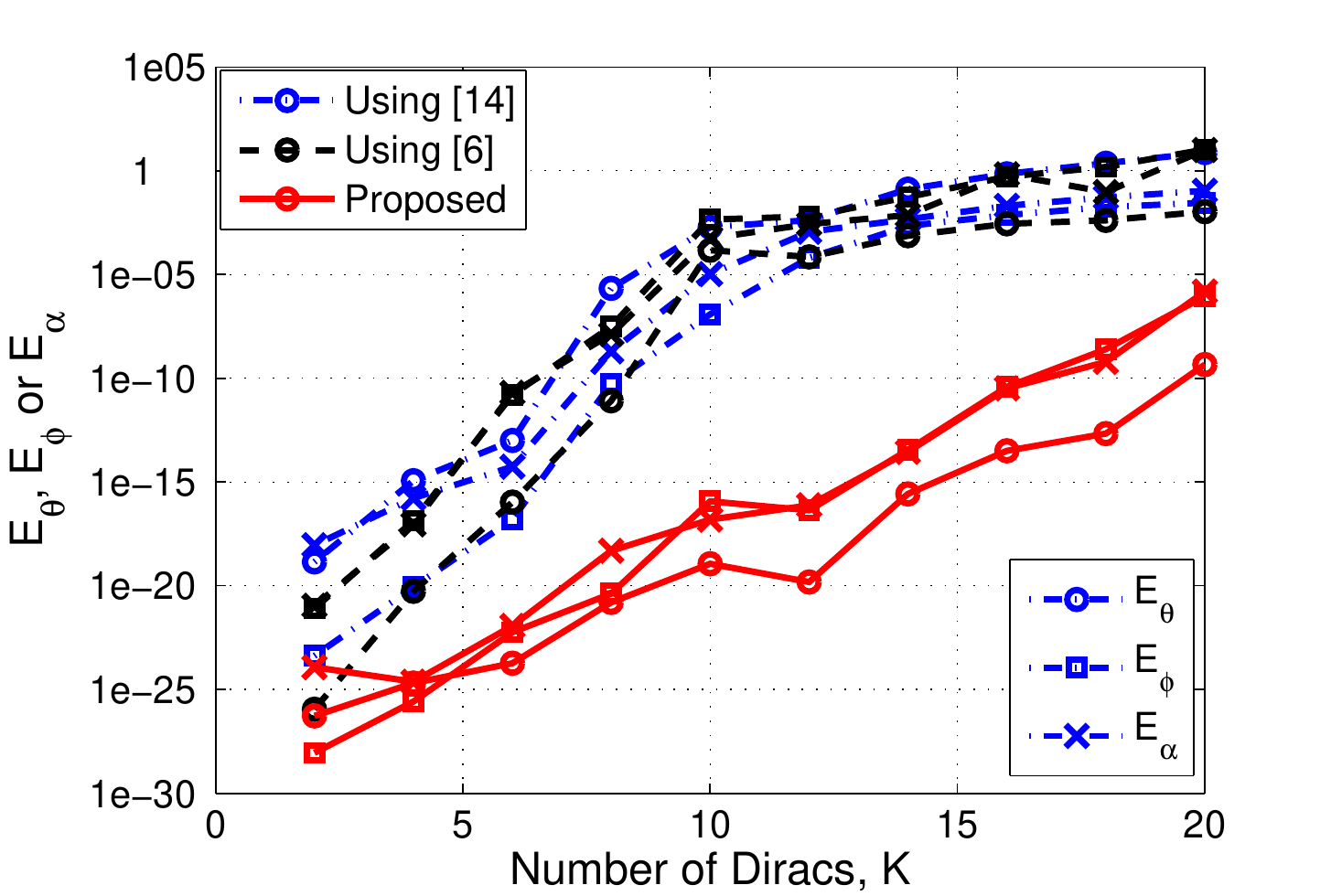}
    \caption{Mean Error $E_\theta$, $E_\phi$ and $E_\alpha$ between recovered and original colatitudes, longitudes and amplitudes respectively for different values of $2 \le K \le 20$~(number of Diracs) using the existing techniques~\cite{Sameul:2013,Dokmanic:2016} and the proposed algorithm. }
    \label{fig:errors}
\end{figure}

\section{Conclusions}\label{sec:conclusion}

In this work, we have proposed a method for accurate reconstruction of an FRI signal consisting of $K$ Dirac functions on the sphere. The proposed method takes samples of the signal bandlimited in the SH domain at the SH degree $L=\lceil K+\sqrt{K+\frac{1}{4}} - \frac{1}{2}\rceil$ for the computation of SH transform of the signal. Following the computation of SH coefficients, we recover the parameters of the Diracs using the annihilating filter method, root finding and solving a series of linear systems. The proposed method requires the same or less number of samples compared to the best of existing methods.
More importantly, the error in the recovery of parameters is significantly smaller. 

%
%
%




%

\bibliography{sht_bib} 

\end{document}